\documentclass[journal]{IEEEtran}
\usepackage{amsmath,amssymb,mathtools,bm} 
\usepackage{caption}
\usepackage{graphicx}
\usepackage{subcaption} %
\usepackage{cite}
\usepackage{float}
\usepackage{graphicx}
\usepackage{booktabs,multirow,array}    
\usepackage{graphicx}
\usepackage{algorithm}        
\usepackage{algorithmic}      
\usepackage{enumitem} 
\usepackage[hidelinks]{hyperref}        \usepackage{xcolor}   
\usepackage{stfloats}         
\usepackage{tikz}
\usetikzlibrary{arrows.meta}
\usepackage{amsmath}

\usepackage{url}
\usepackage{xcolor}
\usepackage{hyperref}

\hypersetup{
  colorlinks=true,
  citecolor=blue,   
  urlcolor=blue,
  linkcolor=green!60!black 
}

\makeatletter
\AtBeginDocument{%
  \let\HyOrg@refstepcounter\refstepcounter
  \def\refstepcounter#1{%
    \HyOrg@refstepcounter{#1}%
    \ifx\Hy@param\@empty\else
      \ifx\Hy@param\Hy@equationstring
        \hypersetup{linkcolor=blue!60!black}%
      \else
        \hypersetup{linkcolor=blue!60!black}%
      \fi
    \fi
  }%
}
\makeatother

\newcommand{\jj}{\mathrm{j}}  
  
\DeclareMathOperator{\diagop}{diag}
\newcommand{\Diag}[1]{\diagop\!\left(#1\right)} 
\newcommand{\D}[1]{\mathrm{diag}\!\left(#1\right)}

\newtheorem{theorem}{Theorem}

\newtheorem{remark}{Remark}

\begin{document}

\title{A Wirtinger Power Flow Jacobian Singularity Condition for Voltage Stability in Converter-Rich Power Systems}

\author{Ahmed Mesfer Alkhudaydi,~\IEEEmembership{Student Member,~IEEE} and
        Bai Cui,~\IEEEmembership{Member,~IEEE}%
\thanks{A. M. Alkhudaydi and B. Cui are with the Department of Electrical and Computer Engineering, Iowa State University, Ames, IA 50011 USA (e-mails: Ahmed92@iastate.edu, baicui@iastate.edu).
}
}

\maketitle

\begin{abstract}
The progression of modern power systems towards converter-rich operations calls for new models and analytics in steady-state voltage stability assessment. The classic modeling assumption of the generators as stiff voltage sources no longer holds. Instead, the voltage- and current-limited behaviors of converters need to be considered. In this paper, we develop a Wirtinger derivative-based formulation for the power flow Jacobian and derive an explicit sufficient condition for its singularity. Compared to existing works, we extend the explicit sufficient singularity condition to incorporate all bus types instead of only slack and PQ types. We prove that the singularity of the alternative Jacobian coincides with that of the conventional one. A bus-wise voltage stability index, denoted $C_{\mathrm{W}}$, is derived from diagonal dominance conditions. The condition \(\min_i C_{W,i}\) being greater than one certifies the nonsingularity of the Jacobian and provides a fast, non-iterative stability margin. Case studies in standard IEEE test systems show that the proposed index yields less conservative and more localized assessments than classical indices such as the L-index, the $K_{\mathrm{R}}$ index, and the SCR index.

\end{abstract}

\begin{IEEEkeywords}
Voltage stability, Wirtinger calculus, Complex Jacobian-based, power flow solvability, Inverter-based resources (IBRs), diagonal dominance.
\end{IEEEkeywords}

\section{Introduction}
\IEEEPARstart{T}{he} rapid integration of inverter-based resources 
(IBRs) has fundamentally reshaped voltage stability mechanisms in 
modern power systems~\cite{Hatziargyriou2021_Definition, Lasseter2020}. As converter penetration increases, steady-state voltage behavior transitions from synchronous-generator dominance to regimes governed by Thevenin equivalent impedance at the points of common coupling (PCCs), converter current limits, and control-mode switching. These challenges are explicitly reflected in emerging grid codes such as IEEE~Standard~2800--2022~\cite{IEEE2800}, which  underscores the need for new voltage stability assessment tools applicable to power systems with high penetration of IBRs.

The classical framework for the stability of the power system, first established by 
the IEEE--CIGR\'E joint task force in 2004~\cite{Kundur2004} and recently revised to address power-electronic phenomena~\cite{Hatziargyriou2021_Definition}, recognizes voltage stability as a core domain alongside angle and frequency stability~\cite{Ge2022_VSG,LaivaRoca2022_FrequencyResponse,Ranjan2021_HybridDVR}. Voltage stability can be broadly classified into large-disturbance (transient) stability, which describes the dynamic voltage response  under severe disturbances, and small-disturbance (quasi-steady-state or static) stability, which characterizes the steady-state sensitivity of bus voltages to incremental perturbations~\cite{Hatziargyriou2021_Definition,VanCutsemVournas1998}. This paper focuses on the latter, specifically the solvability of quasi-steady-state power flow equations under incremental stress.

Traditional static voltage stability tools that include continuation power flow (CPF)~\cite{ajjarapu1992cpf}, Jacobian singularity analysis, and minimum singular value indices provide rigorous system-level margins near the $P$-$V$ nose point~\cite{VanCutsemVournas1998,kundur1994stability}. However, these methods are computationally intensive and offer limited insight into bus-level vulnerability, particularly when converter buses operate under voltage or current magnitude constraints~\cite{Cespedes2015,Zhong2019}. Bus voltage-based indices such as the $L$-index~\cite{KesselGlavitsch1986} 
and impedance-based indices such as the multi-infeed impedance index $K_{\mathrm{R}}$~\cite{Wang2024_Impedance}, and SCR-based metrics including site-dependent SCR (SDSCR)~\cite{Wu2018_SDSCR} and classical SCR~\cite{Wu2018_SDSCR,Lan2025_LPSCR}, provide faster assessments but do not explicitly model converter control modes and operating constraints; classical SCR further suffers from the fact that its Th\'evenin-based formulation absorbs AC loads into the equivalent source, which distorts stability indications in converter-rich, load-proximate networks~\cite{Lan2025_LPSCR}. In prior work~\cite{WangCuiWang2017}, a Wirtinger-calculus-based solvability condition was introduced, which yields less conservative certificates than classical real variable formulations and enables bus-wise stability assessment. However, the formulation was restricted to PQ buses and does not account for voltage-regulated or current-limited converter operation, which are important operating modes in IBR-dominated systems.

This paper addresses this gap by extending the C-index of~\cite{WangCuiWang2017} to a unified bus-type-aware voltage stability index, denoted $C_{\mathrm{W}}$, for systems with mixed unconstrained, voltage-regulated, and current-limited buses. The development has two layers. First, we derive a reduced Wirtinger Jacobian and prove that its singular set coincides with that of the conventional power flow Jacobian, preserving the classical voltage collapse boundary. Second, we use a strict row-diagonal-dominance condition on the reduced Jacobian to construct a computationally efficient per-bus solvability index. In this way, the proposed framework combines exact Jacobian-level equivalence results with a practical sufficient solvability certificate that is suitable for fast screening. The main contributions of this paper are as follows:
\begin{enumerate}
    \item A unified Wirtinger-based reformulation of the steady-state power flow equations, in which voltage and current limits are embedded as tangent-subspace constraints, yields a compact reduced Jacobian $J_{\mathrm{red}}$.
    \item A rigorous proof that the singularity sets of the conventional real-valued Jacobian and the reduced Wirtinger Jacobian coincide, which establishes formal equivalence with the classical voltage collapse boundary.
    \item The per-bus solvability margin $C_{\mathrm{W},i}$, derived from the row-diagonal dominance of $J_{\mathrm{red}}$, provides a sufficient condition for power flow solvability under mixed bus constraints.
\end{enumerate}

The proposed index complements CPF-based margin estimation by enabling fast, bus-level screening of voltage stability using steady-state phasor measurements together with the network impedance matrix, which makes it well suited for real-time monitoring and operational decision support in converter-rich power systems.

The remainder of this paper is organized as follows. Section~\ref{sec:model} introduces the system modeling and Wirtinger preliminaries. Section~\ref{sec:tangent-framework} presents the unified tangent framework for voltage and current constraints. Section~\ref{sec:reduced-jacobian} derives the reduced Wirtinger Jacobian, and Section~\ref{sec:multi-bus} generalizes the formulation and defines the proposed index $C_{\mathrm{W}}$. Section~\ref{sec:case_study} presents numerical case studies and discusses the results, and Section~\ref{sec:Conclusion} concludes the paper.
\section{System Modeling and Wirtinger Preliminaries}
\label{sec:model}

We consider a per-phase network model with one slack bus and $n_u{+}n_c$ buses at the points of common coupling (PCCs), comprising $n_u$ unconstrained buses and $n_c$ constrained buses. Distributed generators, loads, and other converter-interfaced devices at the PCCs are represented through their net power injections. The slack bus provides the system reference phasor.

\subsection{Th\'evenin Form and Bus Partitioning}

Let $\mathbf{Y} \in \mathbb{C}^{(n_u+n_c+1)\times(n_u+n_c+1)}$ denote the bus admittance matrix. Eliminating the slack bus yields the reduced Th\'evenin form 
\begin{equation}
\mathbf{V} = \mathbf{E} + \mathbf{Z}\mathbf{I},
\qquad 
\mathbf{Z} = \mathbf{Y}_{\mathcal{NN}}^{-1}, 
\quad 
\mathbf{E} = -\mathbf{Y}_{\mathcal{NN}}^{-1}\mathbf{Y}_{\mathcal{N}s}V_s ,
\label{eq:thevenin}
\end{equation}
where $\mathbf{V} \in \mathbb{C}^{n_u+n_c}$ and $\mathbf{I} \in \mathbb{C}^{n_u+n_c}$ are the PCC voltage and current injection vectors, $\mathbf{Z}$ is the reduced impedance matrix, and $\mathbf{E}$ is the equivalent source voltage determined by the slack bus. The bus set is partitioned as
\begin{equation*}
\mathcal{N} = \mathcal{U} \cup \mathcal{C}, 
\quad 
\mathcal{U} = \{1,\dots,n_u\}, \quad 
\mathcal{C} = \{n_u+1,\dots,n_u+n_c\},
\end{equation*}
with $\mathcal{U}$ denoting unconstrained buses and $\mathcal{C}$ denoting magnitude-constrained buses. For an unconstrained bus $i \in \mathcal{U}$, the complex power injection is
\begin{equation}
S_i = V_i I_i^{*}, 
\qquad 
S_i = P_i + jQ_i .
\label{eq:unconstrained-power}
\end{equation}
For a constrained bus $j \in \mathcal{C}$, the active power is specified while either voltage or current magnitude is enforced:
\begin{equation}
P_j = \Re\{ V_j I_j^{*} \}, 
\qquad 
|V_j| = \text{const.} \;\text{ or }\; |I_j| = \text{const.},
\label{eq:constrained-power}
\end{equation}
where the corresponding reactive power $Q_j$ (or current phase) is implicitly determined by the power flow solution to satisfy the constraint.  

\subsection{Wirtinger Calculus Preliminaries}

Voltage stability assessment requires evaluating the sensitivity of complex power with respect to phasor perturbations. Since complex conjugation is not holomorphic (i.e., the standard complex derivative does not exist for functions of $z$ and $z^*$ simultaneously), Wirtinger calculus treats a complex variable and its conjugate as independent analytic coordinates~\cite{WangCuiWang2017}.  
For a complex-valued function $f(z,z^{*})$ with $z = x + jy$, the Wirtinger derivatives are defined by
\begin{equation}
\frac{\partial f}{\partial z} 
= \frac{1}{2} \left( \frac{\partial f}{\partial x} - j \frac{\partial f}{\partial y}\right),
\qquad
\frac{\partial f}{\partial z^{*}} 
= \frac{1}{2}\left( \frac{\partial f}{\partial x} + j \frac{\partial f}{\partial y}\right).
\label{eq:wirtinger}
\end{equation}
This formulation enables the direct evaluation of complex sensitivities in phasor form by treating any complex variable and its conjugate as independent analytic coordinates~\cite{dvzafic2018high}.

\subsection{Full Wirtinger State and Jacobian Formulation}

The full Wirtinger state and function vectors are defined by stacking complex power injections at unconstrained buses and real power injections at constrained buses:
\begin{equation*}
\underbrace{
\mathbf{x}_{\mathrm{full}}
=
\begin{bmatrix} 
\mathbf{I}_{\mathcal{U}} \\[2pt]
\mathbf{I}_{\mathcal{U}}^{*} \\[2pt]
\mathbf{I}_{\mathcal{C}} \\[2pt]
\mathbf{I}_{\mathcal{C}}^{*}
\end{bmatrix}
\in\mathbb{C}^{2n_u+2n_c}
}_{\text{full state}},
\qquad
\underbrace{
\mathbf{F}_{\mathrm{full}}
=
\begin{bmatrix} 
\mathbf{S}_{\mathcal{U}} \\[2pt]
\mathbf{S}_{\mathcal{U}}^{*} \\[2pt]
\mathbf{P}_{\mathcal{C}} \\[2pt]
\mathbf{P}_{\mathcal{C}}^{*}
\end{bmatrix}
\in\mathbb{C}^{2n_u+2n_c}
}_{\text{full function}}.
\end{equation*}

The resulting full Wirtinger Jacobian $\mathbf{J}_{\mathrm{full}} \in \mathbb{C}^{(2n_u+2n_c)\times(2n_u+2n_c)}$ contains all first-order power sensitivities with respect to the complex current coordinates

\begin{equation}
\mathbf{J}_{\mathrm{full}}
=
\frac{\partial \mathbf{F}_{\mathrm{full}}}{\partial \mathbf{x}_{\mathrm{full}}}
=
\begin{bmatrix}
\frac{\partial \mathbf{S}_{\mathcal{U}}}{\partial \mathbf{I}_{\mathcal{U}}} &
\frac{\partial \mathbf{S}_{\mathcal{U}}}{\partial \mathbf{I}_{\mathcal{U}}^{*}} &
\frac{\partial \mathbf{S}_{\mathcal{U}}}{\partial \mathbf{I}_{\mathcal{C}}} &
\frac{\partial \mathbf{S}_{\mathcal{U}}}{\partial \mathbf{I}_{\mathcal{C}}^{*}} \\[4pt]
\frac{\partial \mathbf{S}_{\mathcal{U}}^{*}}{\partial \mathbf{I}_{\mathcal{U}}} &
\frac{\partial \mathbf{S}_{\mathcal{U}}^{*}}{\partial \mathbf{I}_{\mathcal{U}}^{*}} &
\frac{\partial \mathbf{S}_{\mathcal{U}}^{*}}{\partial \mathbf{I}_{\mathcal{C}}} &
\frac{\partial \mathbf{S}_{\mathcal{U}}^{*}}{\partial \mathbf{I}_{\mathcal{C}}^{*}} \\[4pt]
\frac{\partial \mathbf{P}_{\mathcal{C}}}{\partial \mathbf{I}_{\mathcal{U}}} &
\frac{\partial \mathbf{P}_{\mathcal{C}}}{\partial \mathbf{I}_{\mathcal{U}}^{*}} &
\frac{\partial \mathbf{P}_{\mathcal{C}}}{\partial \mathbf{I}_{\mathcal{C}}} &
\frac{\partial \mathbf{P}_{\mathcal{C}}}{\partial \mathbf{I}_{\mathcal{C}}^{*}} \\[4pt]
\frac{\partial \mathbf{P}_{\mathcal{C}}^{*}}{\partial \mathbf{I}_{\mathcal{U}}} &
\frac{\partial \mathbf{P}_{\mathcal{C}}^{*}}{\partial \mathbf{I}_{\mathcal{U}}^{*}} &
\frac{\partial \mathbf{P}_{\mathcal{C}}^{*}}{\partial \mathbf{I}_{\mathcal{C}}} &
\frac{\partial \mathbf{P}_{\mathcal{C}}^{*}}{\partial \mathbf{I}_{\mathcal{C}}^{*}}
\end{bmatrix}.
\label{eq:fulljac}
\end{equation}

The singularity of $\mathbf{J}_{\mathrm{full}}$ is analyzed 
in Sec.~\ref{sec:tangent-framework}, where the constrained bus tangent relations are used to eliminate dependent columns and yield the reduced Wirtinger Jacobian.


\section{Unified Tangent-vector Framework for Bus Constraint Modeling}
\label{sec:tangent-framework}

In converter-dominated power systems, control and protection mechanisms impose magnitude constraints at the PCCs through their control and protection systems. 
For example, a PV bus maintains a prescribed voltage magnitude through reactive current control, whereas a converter operating under current-limited conditions enforces a fixed current magnitude. 
In both cases, admissible perturbations in the complex-current vector are confined to a one-dimensional tangent subspace representing directions of permissible motion that preserve the magnitude constraint.

The proposed formulation represents both voltage- and 
current-magnitude constraints as constant magnitude loci 
in the complex plane. Within this framework, the differential relationship between a complex variable and its conjugate coordinate is expressed through a unit-modulus tangent factor that enforces motion along the admissible manifold. This section develops a unified mathematical representation of these magnitude-constrained operating loci within the Wirtinger framework, thereby enabling a consistent differential formulation applicable across all bus types.

\subsection{Concept of Tangent Constraints}
In conventional unconstrained operation, a converter bus  has two independent real degrees of freedom, i.e., the real and imaginary parts of the complex current $I_i$ or the complex voltage $V_i$. Under a fixed-magnitude constraint (\(|V|=\mathrm{const}\) or \(|I|=\mathrm{const}\)), admissible perturbations \((dV,\,dI)\) must satisfy the corresponding magnitude condition. These constraints define a one-dimensional subspace in the complex plane along which the state of the system may evolve without violating the constraint. The admissible direction of differential motion within this subspace is parameterized by a complex scalar of unit modulus, which captures the allowable phase variation that preserves the constraint~\cite{sym15040865,poongothai2021novel,chu2023stability}.

\subsubsection{Voltage Constraint Bus}
In a voltage-limited bus $i$, the voltage magnitude remains constant, therefore $d|V_i|=0$. 
Differentiating $|V_i|^2 = V_i V_i^{*}$ gives
\begin{equation}
V_i^{*}\,dV_i + V_i\,dV_i^{*} = 0.
\label{eq:pv-diff}
\end{equation}
Taking into account the local self-sensitivity through the diagonal term of the impedance matrix with the network relation $d\mathbf V = \mathbf Zd\mathbf I$, we have $dV_i = Z_{ii}dI_i$ and $dV_i^{*} = Z_{ii}^{*}dI_i^{*}$.  
Substituting into~\eqref{eq:pv-diff} yields the tangent relation
\begin{equation}
dI_i^{*} = \kappa_i\, dI_i, 
\qquad 
\kappa_i \triangleq -\,\frac{V_i^{*} Z_{ii}}{V_i Z_{ii}^{*}},
\label{eq:kappa-def}
\end{equation}
where $|\kappa_i|=1$.

The unit-modulus factor $\kappa_i$ rotates $dI_i^{*}$ 
to preserve $|V_i|$ while allowing phase variation.

\subsubsection{Current Constraint Bus}
Under current-limited operation, the magnitude of the injected current is fixed, i.e., $|I_i| = I_{i,\max}=\mathrm{const}$. 
Differentiating $|I_i|^2 = I_i I_i^{*}$ gives

\begin{equation}
I_i^{*}\,dI_i + I_i\,dI_i^{*} = 0
\;\;\Longrightarrow\;\;
dI_i^{*} = \zeta_i\, dI_i, 
\quad 
\zeta_i \triangleq -\,\frac{I_i^{*}}{I_i}
\label{eq:zeta-def}
\end{equation}
where $|\zeta_i|=1$ defines the rotation of the conjugate current differential required to preserve the constant-current magnitude. 
As in the voltage-constrained case, admissible increments lie along a tangent direction orthogonal to the current vector; this preserves its magnitude while allowing for phase adjustment. 
This is consistent with the current-limiting requirements 
of IEEE~Std.~2800--2022~\cite{IEEE2800}.

Equations~\eqref{eq:kappa-def} and~\eqref{eq:zeta-def} show that both voltage- and current-magnitude constraints restrict admissible perturbations to a one-dimensional tangent subspace in the complex plane. Figs.~\ref{fig:constV_tangent} and~\ref{fig:constI_tangent} 
illustrate the geometric interpretation for the constant-voltage and constant-current cases, respectively. In Fig.~\ref{fig:constV_tangent}, $V_i$ lies on the circle $|V_i| = \mathrm{const}$, and the unit-modulus factor $\kappa_i$ defines the tangent direction satisfying $dI_i^{*} = \kappa_i dI_i$, which preserves $|V_i|$. 
In Fig.~\ref{fig:constI_tangent}, $I_i$ lies on the circle 
$|I_i| = \mathrm{const}$, and the rotation $\zeta_i$ defines the tangent relation $dI_i^{*} = \zeta_i dI_i$, which ensures $d|I_i| = 0$.

\begin{figure}[t]
  \centering
  \includegraphics[width=0.7\columnwidth]{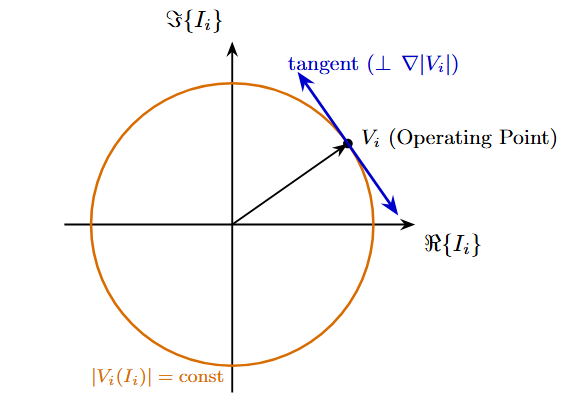}
  \caption{Tangent direction on the constant-$|V|$ space.}
  \label{fig:constV_tangent}
\end{figure}

\begin{figure}[t]
  \centering
  \includegraphics[width=0.7\columnwidth]{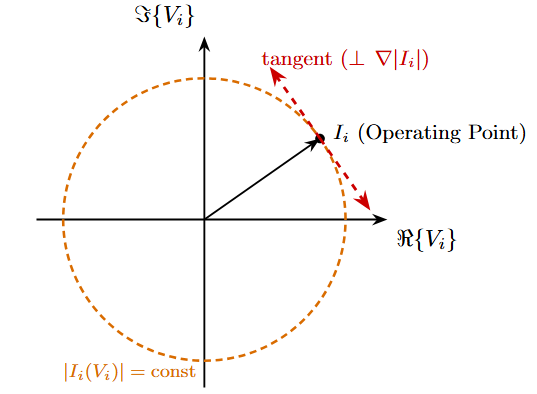}
  \caption{Tangent direction on the constant-$|I|$ space.}
  \label{fig:constI_tangent}
\end{figure}

\subsection{Generic Tangent Representation}
The voltage and current magnitude constraints admit a unified representation in the complex plane. For each bus~$i$, admissible perturbations satisfy the tangent condition 
\begin{equation}
dI_i^{*} = \xi_i\, dI_i, 
\qquad 
|\xi_i| = 1, \quad i \in \mathcal{C},
\label{eq:xi-generic1}
\end{equation}
where the tangent factor $\xi_i$ is defined via
\begin{equation}
\xi_i =
\begin{cases}
0, & i \in \mathcal{U} \ \text{(unconstrained bus)}, \\[3pt]
\kappa_i, & i \in \mathcal{C}^V \ \text{(voltage-constrained bus)}, \\[3pt]
\zeta_i, & i \in \mathcal{C}^I \ \text{(current-constrained bus)},
\end{cases}
\label{eq:xi-cases}
\end{equation}
with $\mathcal{C}^{V}$ and $\mathcal{C}^{I}$ denoting voltage- and current-constrained buses, respectively.

Geometrically, $\xi_i$ encodes the direction of admissible 
perturbation on the constraint circle in the complex plane. For unconstrained buses ($\xi_i = 0$), the Wirtinger derivatives remain independent. For voltage-constrained buses ($\xi_i = \kappa_i$), perturbations remain tangential to the constant-voltage space, which enforces $d|V_i| = 0$. Similarly, for current-constrained buses ($\xi_i = \zeta_i$), motion is restricted to the constant-current space, which enforces $d|I_i| = 0$. In all constrained cases, the unit modulus $|\xi_i| = 1$ ensures that the perturbation direction preserves the magnitude constraint. The tangent relation~\eqref{eq:xi-generic1} thus provides a single analytical form applicable to all bus types. When applied to linearized power flow equations, it projects the full Wirtinger Jacobian onto the admissible subspace, which forms the basis for the reduced Jacobian $J_{\mathrm{red}}$ derived in Sec.~\ref{sec:reduced-jacobian}.

\section{Reduced Wirtinger Jacobian and Two-Bus example}
\label{sec:reduced-jacobian}
This section derives the reduced Wirtinger Jacobian by systematically applying the tangent-constraint relations to the full Wirtinger formulation. The resulting reduced matrix explicitly captures the coupling between complex power and current variations under voltage or current magnitude constraints. To illustrate its analytical form and confirm the structural consistency, a two-bus system is used as an illustrative example which provides intuitive verification of the reduction process.

\subsection{Reduced Wirtinger Jacobian}
For any magnitude-limited bus $i$, whether voltage or current constrained, the real power injection is expressed as
\begin{equation}
P_i = \tfrac{1}{2}\big(V_i I_i^{*} + V_i^{*} I_i\big).
\label{eq:power-bus}
\end{equation}

Using the local Th\'evenin relation $V_i = E_i + Z_{ii} I_i$ and treating $I_i$ and $I_i^{*}$ as independent Wirtinger variables, the partial derivatives are
\begin{equation}
\frac{\partial P_i}{\partial I_i}
= \tfrac{1}{2}\big(V_i^{*} + Z_{ii} I_i^{*}\big),
\qquad
\frac{\partial P_i}{\partial I_i^{*}}
= \tfrac{1}{2}\big(V_i + Z_{ii}^{*} I_i\big).
\label{eq:partials-generic}
\end{equation}

Define the auxiliary complex scalar as
\begin{equation}
\alpha_i \triangleq \tfrac{1}{2}\big(V_i + Z_{ii}^{*} I_i\big),
\label{eq:alpha-def}
\end{equation}
Therefore, 
\begin{equation}
\frac{\partial P_i}{\partial I_i} = \alpha_i^{*},
\qquad
\frac{\partial P_i}{\partial I_i^{*}} = \alpha_i.
\label{eq:alpha-partials}
\end{equation}

Applying the tangent relation in~\eqref{eq:xi-generic1}, where
$\xi_i = \kappa_i$ for voltage-limited buses and $\xi_i = \zeta_i$ for current-limited buses,
the total differential becomes:
\begin{equation}
dP_i 
= 
\Big(
\frac{\partial P_i}{\partial I_i}
+ 
\xi_i\, \frac{\partial P_i}{\partial I_i^{*}}
\Big) dI_i
= (\alpha_i^{*} + \xi_i \alpha_i)\, dI_i.
\label{eq:differential-generic}
\end{equation}
The corresponding element of $J_{\mathrm{red}}$ is therefore 
\begin{equation}
J_{\mathrm{red}}^{(i)} = \alpha_i^{*} + \xi_i \alpha_i,
\label{eq:Jred-generic}
\end{equation}
where $\xi_i$ encodes the type of magnitude constraint and its associated tangent rotation.

The complex scalar $\alpha_i$ represents the local coupling between 
bus voltage and current through the Th\'evenin impedance, while $\kappa_i$ 
defines the rotation ensuring constant-voltage operation.  
The product $(\alpha_i^{*} + \kappa_i \alpha_i)$ quantifies the effective 
tangent sensitivity between incremental current and active power. As $|J_{\mathrm{red}}^{(i)}|$ approaches zero, the magnitude-constrained operating locus becomes tangent to the power flow feasibility boundary, thereby defining the local solvability limit.

\subsection{Two-Bus Case}
\label{subsec:two-bus-validation}





\begin{figure}[!t]
\centering
\begin{tikzpicture}[>=latex]
  \draw (0,0) circle (0.4);
  \node at (0,0) {$\sim$};
  \draw (0.4,0) -- (1,0);

  \draw[thick] (1,-0.9) -- (1,0.9);
  \node[above] at (1,0.9) {$E_s \angle 0^\circ$};
  \node[below, align=center] at (1,-0.9) {\textbf{1}};

  \draw (1,0) -- (2,0);
  \draw[fill=white] (2,-0.2) rectangle (3.5,0.2);
  \node[above] at (2.75,0.2) {$Z = R + jX$};
  \draw (3.5,0) -- (4.2,0);

  \draw[thick] (4.2,-0.9) -- (4.2,0.9);
  \node[above] at (4.2,0.9) {$V_2 \angle \theta_2$};
  \node[below, align=center] at (4.2,-0.9) {\textbf{2}};

  \draw (4.2,0) -- (4.6,0);      
  \draw (4.6,-0.4) rectangle (5.4,0.4);

  \draw (4.6,0.4) -- (5.4,-0.4);
  \draw (4.95,0.26) -- (5.25,0.26);
  \draw (4.95,0.16) -- (5.25,0.16);

  \draw (4.75,-0.22)
        .. controls (4.87,-0.02) and (4.98,-0.29) ..
        (5.10,-0.12);

  \node[below] at (5.0,-0.4) {IBR};
\end{tikzpicture}
\caption{Example of a two bus system.}
\label{fig:twobus}
\end{figure}
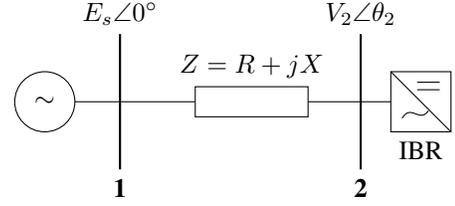
  

The two-bus system in Fig.~\ref{fig:twobus} consists of a slack bus (bus~1) with fixed voltage $|E_s|$ and a 
voltage-constrained bus (bus~2) with fixed $|V_2|$. In bus 2, the voltage magnitude $|V_2|$ is fixed via the generator or the converter control system. This constraint removes one of the radial degrees of freedom of the complex voltage because the bus voltage can only move tangentially on the complex plane, which must stay on the circle.  
Consequently, only the phase angle~$\theta$ can vary, which represents the single independent state variable that governs small perturbations in that bus.  Accordingly, the reduced Jacobian for the two-bus system collapses to a scalar element that relates incremental active power changes to phase angle variations as 

\begin{equation}
J_{\mathrm{red}}^{(2\text{-bus})}
= \alpha_{2}^{*} + \kappa_{2} \alpha_{2},
\label{eq:Jred-2bus}
\end{equation}
where $\alpha_{2}$ is the self-sensitivity coefficient of the reduced network model defined in~\eqref{eq:alpha-def}, and $\kappa_{2}$ is the unit-modulus tangent factor defined in~\eqref{eq:kappa-def} by the voltage constraint. Thus, this direction spans a one-dimensional tangent subspace, in contrast to the two-dimensional space available at unconstrained buses.

The singularity condition
\(
\det(J_{\mathrm{red}}^{(2\text{-bus})}) = 0
\), identifies the local boundary of voltage stability. For comparison, the conventional real-variable Jacobian for a two-bus system with a PV bus contains only the active power sensitivity with respect to the voltage angle 
\begin{equation}
J_{\mathrm{conv}} = 
\left[\frac{\partial P_{2}}{\partial \theta_{2}}\right],
\label{eq:Jconv-PV}
\end{equation}
and the corresponding incremental mismatch relation is 
\begin{equation}
\Delta P_{2} = 
\left(\frac{\partial P_{2}}{\partial \theta_{2}}\right)\Delta \theta_{2}.
\end{equation}
At the $P$–$V$ nose point, $\partial P_{2}/\partial \theta_{2} = 0$, indicating the same singular condition as $J_{\mathrm{red}}^{(2\text{-bus})}=0$. Hence, the reduced Wirtinger Jacobian and the conventional Jacobian share the same singularity condition
\begin{multline}
J_{\mathrm{conv}} = L\,J_{\mathrm{red}}^{(2\text{-bus})}R,
\quad
\det(J_{\mathrm{conv}})=0 \\
\iff 
\det(J_{\mathrm{red}}^{(2\text{-bus})})=0,
\end{multline}
where $L$ and $R$ are non-singular real transformation factors.  
This confirms that both formulations identify the same 
solvability boundary. The general multi-bus equivalence 
is established in Appendix~\ref{app:eqv}, where the Wirtinger formulation is shown to provide a geometrically 
consistent complex-domain representation that inherently respects voltage magnitude 
constraints~\cite{kundur1994stability,VanCutsemVournas1998,
sauer1990jacobian,ajjarapu1992cpf,gao1992modal}.


\section{Multi-Bus Extension and Unified Stability Index}
\label{sec:multi-bus}

Starting from the full Wirtinger Jacobian $\mathbf{J}_{\mathrm{full}}$ in~\eqref{eq:fulljac}, 
the tangent-constrained differential relation
\begin{equation}
d\mathbf{I}^{*} = \boldsymbol{\Xi}\, d\mathbf{I},
\qquad
\boldsymbol{\Xi} = \mathrm{diag}(\xi_1,\xi_2,\ldots,\xi_{n_u+n_c}),
\label{eq:xi-generic}
\end{equation}
restricts admissible complex perturbations at each constrained bus to a single complex degree of freedom.  
This relation may be expressed in current coordinates ($d\mathbf{I}$) or, equivalently, in voltage coordinates ($d\mathbf{V}$), since both representations are linked by the network relation in equation~\eqref{eq:thevenin}.  
For unconstrained buses ($i \in \mathcal{U}$), $dI_i$ and $dI_i^{*}$ (or $dV_i$ and $dV_i^{*}$) remain independent, corresponding to $\xi_i = 0$.  
For constrained buses ($j \in \mathcal{C}$), $\xi_j$ has a unit modulus and defines the local tangent rotation that enforces the constant-magnitude constraint, which is either constant-$|V|$ or constant-$|I|$. The total differential of the full complex power mismatch function is
\begin{equation}
d\mathbf{F}_{\mathrm{full}}
=
\frac{\partial \mathbf{F}_{\mathrm{full}}}{\partial \mathbf{I}}\, d\mathbf{I}
+
\frac{\partial \mathbf{F}_{\mathrm{full}}}{\partial \mathbf{I}^{*}}\, d\mathbf{I}^{*}.
\label{eq:dFfull}
\end{equation}
Substituting~\eqref{eq:xi-generic} gives the reduced differential mapping
\begin{equation}
d\mathbf{F}_{\mathrm{red}}
=
\Big(
\frac{\partial \mathbf{F}_{\mathrm{full}}}{\partial \mathbf{I}}
+
\frac{\partial \mathbf{F}_{\mathrm{full}}}{\partial \mathbf{I}^{*}}\boldsymbol{\Xi}
\Big)d\mathbf{I}
\;\triangleq\;
\mathbf{J}_{\mathrm{red}}\, d\mathbf{I},
\label{eq:Jred-def}
\end{equation}
where $\mathbf{J}_{\mathrm{red}}$ is the reduced Wirtinger Jacobian.

Equation~\eqref{eq:Jred-def} shows that for each constrained bus~$i$, the two Wirtinger columns associated with $I_i$ and $I_i^{*}$ are projected on a single tangent direction according to the following
\begin{equation}
\frac{\partial (\cdot)}{\partial I_i}
\;\longrightarrow\;
\frac{\partial (\cdot)}{\partial I_i}
+
\xi_i\, \frac{\partial (\cdot)}{\partial I_i^{*}}.
\label{eq:column-merge}
\end{equation}
This column-merging operation projects the sensitivity matrix onto the feasible tangent subspace defined by the constraint $dI_i^{*} = \xi_i\, dI_i$, which reduces the column count without affecting the number of rows.  

When the voltage or current magnitude at a constrained bus is fixed, the tangent relation enforces a conjugate symmetry between the differential power components.  
As a result, one of the two conjugate equations in the full Wirtinger Jacobian becomes algebraically dependent on its counterpart.  
This dependence means that one of the two rows corresponding to each constrained bus becomes redundant and can therefore be eliminated.  
Hence, the elimination of the rows (one per constrained bus) logically follows from the algebraic dependence created by the tangent constraint.

After merging columns and eliminating conjugate dependent rows, the reduced Jacobian $\mathbf{J}_{\mathrm{red}}$ retains only independent sensitivities while remaining physically consistent. In the full Wirtinger Jacobian $\mathbf{J}_{\mathrm{full}}$, each bus contributes two rows and two columns (for $I_i$ and $I_i^{*}$). Under tangent reduction, the unconstrained buses ($\xi_i=0$) keep both rows and both columns, while the constrained buses ($|\xi_i|=1$) have their two columns combined via~\eqref{eq:column-merge} and a redundant conjugate row removed. Consequently, $\mathbf{J}_{\mathrm{red}}$ has dimension $(2n_u+n_c)\times(2n_u+n_c)$, which comprises two degrees of freedom per unconstrained bus and one per constrained bus, exactly matching the number of independent complex degrees of freedom. Each unit modulus factor $\xi_i$ specifies the local rotation associated with the constant magnitude constraint, and $\mathbf{J}_{\mathrm{red}}$ provides the dimensionally exact mapping of the feasible tangent subspace that supports the subsequent row dominance test and the unified solvability index.
 
Using $\mathbf{V} = \mathbf{E} + \mathbf{Z}\mathbf{I}$, the Wirtinger partial derivatives depend on the bus type. For an unconstrained bus, the complex power injection is $S_i = V_i I_i^{*}$, with derivatives:
\begin{equation}
\left\{
\begin{aligned}
\frac{\partial S_i}{\partial I_i}      &= \Diag{I_i^{*}}\, Z_{ii}, &
\quad \frac{\partial S_i^{*}}{\partial I_i}      &= \Diag{V_i^{*}}, \\[4pt]
\frac{\partial S_i}{\partial I_i^{*}}  &= \Diag{V_i},              &
\quad \frac{\partial S_i^{*}}{\partial I_i^{*}} &= \Diag{I_i}\, Z_{ii}^{*}, \\[4pt]
\frac{\partial S_i}{\partial I_j}      &= \Diag{I_i^{*}}\, Z_{ij}, &
\quad \frac{\partial S_i^{*}}{\partial I_j}     &= 0, \quad j\neq i \\[4pt]
\frac{\partial S_i}{\partial I_j^{*}}  &= 0,                       &
\quad \frac{\partial S_i^{*}}{\partial I_j^{*}} &= \Diag{I_i}\, Z_{ij}^{*}, \quad j\neq i
\end{aligned}
\right.
\label{eq:general-pq-block}
\end{equation}

where $\mathrm{diag}(\cdot)$ forms a diagonal matrix from its argument.

For buses constrained by voltage and current, the real power injection is $P_i = \frac{1}{2}(V_i I_i^{*} + V_i^{*} I_i)$, with derivatives:
\begin{equation}
\left\{
\begin{aligned}
\frac{\partial P_i}{\partial I_j}     &= \tfrac{1}{2}\,\Diag{I_i^{*}}\, Z_{ij}, \quad j\neq i \\[2pt]
\frac{\partial P_i}{\partial I_j^{*}} &= \tfrac{1}{2}\,\Diag{I_i}\, Z_{ij}^{*}, \quad j\neq i \\[2pt]
\frac{\partial P_i}{\partial I_i} &= \tfrac{1}{2}\!\Big(\Diag{I_i^{*}}\, Z_{ii} + \Diag{V_i^{*}}\Big), 
\\[2pt]
\frac{\partial P_i}{\partial I_i^{*}} &= \tfrac{1}{2}\!\Big(\Diag{V_i} + \Diag{I_i}\, Z_{ii}^{*}\Big).
\end{aligned}
\right.
\label{eq:general-pv-groups}
\end{equation}

The reduced Wirtinger Jacobian $\mathbf{J}_{\mathrm{red}}$ 
for the general multi-bus case takes the block form

\begin{equation}
\mathbf{J}_{\mathrm{red}} =
\begin{bmatrix}
\underbrace{\operatorname{diag}(I_i^{*}) Z_{ii}}_{\partial S_i/\partial I_i} & 
\underbrace{\operatorname{diag}(V_i)}_{\partial S_i/\partial I_i^{*}} & 
\underbrace{\operatorname{diag}(I_i^{*}) Z_{ij}}_{\partial S_i/\partial I_j} \\
\underbrace{\operatorname{diag}(V_i^{*})}_{\partial S_i^{*}/\partial I_i} & 
\underbrace{\operatorname{diag}(I_i) Z_{ii}^{*}}_{\partial S_i^{*}/\partial I_i^{*}} & 
\underbrace{\operatorname{diag}(I_i) Z_{ij}^{*} \Xi}_{\partial S_i^{*}/\partial I_j} \\
\underbrace{\tfrac{1}{2} \operatorname{diag}(I_j^{*}) Z_{ij}}_{\partial P_j/\partial I_i} & 
\underbrace{\tfrac{1}{2} \operatorname{diag}(I_j) Z_{ij}^{*}}_{\partial P_j/\partial I_i^{*}} & 
\underbrace{\operatorname{diag}(\alpha_j^{*} + \alpha_j \Xi)}_{\partial P_j/\partial I_j}
\end{bmatrix}
\label{eq:Jred-general}
\end{equation}

By the Lévy--Desplanques theorem~\cite{HornJohnson2013}, 
$\mathbf{J}_{\mathrm{red}}$ is nonsingular if it is strictly row diagonally dominant:
\begin{equation}
|[J]_{ii}| > \sum_{j\neq i}|[J]_{ij}|, \qquad \forall i,
\label{eq:LD-condition}
\end{equation}

Applying this to the reduced Wirtinger Jacobian yields simplified dominance conditions:

\subsubsection{Unconstrained Bus Dominance}
For unconstrained buses, the diagonal dominance condition reduces to
\begin{equation}
|V_i| \;>\; |I_i|\sum_{j \neq i} |Z_{ij}|.
\label{eq:unconstrained-dominance}
\end{equation}

\subsubsection{Constrained Bus Dominance}
For constrained buses, the dominance condition is given as
\begin{equation}
|\alpha_i^{*} + \xi_i \alpha_i| > |I_i|\sum_{j \neq i} |Z_{ij}|,
\label{eq:constrained-dominance}
\end{equation}
where $\alpha_i$ is the local self-sensitivity 
coefficient defined in~\eqref{eq:alpha-def}, and $\xi_i$ is the tangent factor defined 
in~\eqref{eq:xi-generic1}.

\subsection{Unified Solvability Index}

The unified solvability index for bus $i$ is defined as

\begin{equation}
C_{W,i}
\triangleq
\frac{\chi^{\mathrm{U}}_i |V_i| + \chi^{\mathrm{C}}_i |\alpha_i^{*} + \xi_i \alpha_i|}
     {|I_i|\sum_{j \neq i} |Z_{ij}|},
\label{eq:CW-index}
\end{equation}

where $\chi^{\mathrm{U}}_i = 1$ if $i \in \mathcal{U}$ and zero 
otherwise, and $\chi^{\mathrm{C}}_i = 1$ if $i \in \mathcal{C}$ 
and zero otherwise.

If $C_{W,i} > 1$ for all $i \in \mathcal{N}$, then $J_{\mathrm{red}}$ is strictly diagonally dominant and hence nonsingular~\cite{HornJohnson2013}, which ensures a unique local power-flow solution. The system-level index is then defined as
\begin{equation}
C_W = \min_{i} {C_{W,i}}
\label{eq:CW-system}
\end{equation}
which serves as a unified solvability certificate for the entire network. As $C_W$ reaches unity, the system approaches the power-flow solvability boundary.

Each tangent factor $\xi_i$ is a unit-modulus rotation that confines admissible perturbations to the magnitude-constrained operating locus in the complex domain. For unconstrained buses ($i \in \mathcal{U}$), no magnitude constraint is imposed and the current differential $dI_i$ remains free. For constrained buses, $\xi_i$ enforces the appropriate geometric tangent relation (constant-$|V|$ or constant-$|I|$).

The numerator of $C_{W,i}$ captures the local bus strength, where $|V_i|$ for unconstrained buses reflects voltage stiffness, while $|\alpha_i^* + \xi_i\alpha_i|$ for constrained buses 
reflects the active-power sensitivity along the magnitude-constrained operating locus. The denominator $|I_i|\sum_{j\neq i}|Z_{ij}|$ quantifies the aggregate influence of neighboring buses through 
off-diagonal impedance coupling. 

Evaluating $C_W$ requires only the network impedance matrix $\mathbf{Z}$, steady-state phasors $(\mathbf{V}, \mathbf{I})$, and the tangent factors $\xi_i$. The bus-type indicators 
$\chi^{\mathrm{U}}_i$ and $\chi^{\mathrm{C}}_i$ are 
determined directly from the converter operating mode. No real-imaginary decomposition is required, and the index is evaluated directly from network phasors without reformulating the power flow Jacobian. The equivalence between 
$J_{\mathrm{red}}$ and the conventional power flow Jacobian is established in Appendix~\ref{app:eqv}.



\section{Case Studies}
\label{sec:case_study}
The proposed index is evaluated on three test systems of increasing size and complexity: a simple three-bus system, the 9-Bus Hami Region Network, and a modified IEEE 39-bus transmission network.

\subsection{Example of Three-Bus System}
\label{sec:case_study_3-Bus}
A three-bus radial system is used to illustrate the structural behavior of $J_{\mathrm{red}}$ under increasing load. The system parameters are $Z_{12} = 0.080 + j0.400$~p.u., $Z_{23} = 0.100 + j0.500$~p.u., with load profile $S_2 = \lambda(2.00 + j1.00)$~p.u. and a PV bus at Bus~3 ($P_3 = 0.80$~p.u., $V_3 = 1.00$~p.u.).

Fig.~\ref{fig:wirtinger_evolution} shows the evolution of $J_{\mathrm{red}}$ across loading levels. Under light load, the matrix exhibits strong diagonal dominance ($m_{\min} = 0.8431$), as shown in Fig.~\ref{fig:wirtinger_evolution}(a). Near the stability boundary, off-diagonal coupling intensifies, and the diagonal dominance margin drops to $m_{\min} = -0.0326$, as shown in Fig.~\ref{fig:wirtinger_evolution}(b). Diagonal dominance is lost at $\lambda = 0.5974$, which is 0.59\% before the conventional Jacobian singularity at $\lambda = 0.6009$. This 
coincides with the point at which the off-diagonal coupling $\sum_{j \neq i}|Z_{ij}|$ exceeds the local self-impedance $|Z_{ii}|$ at Bus~2, which reflects the saturation of the reactive support at the PV-regulated Bus~3 as the load increases.

\begin{figure}[!t]
\centering
\begin{subfigure}{0.4\textwidth}
    \centering
    \includegraphics[width=\linewidth]{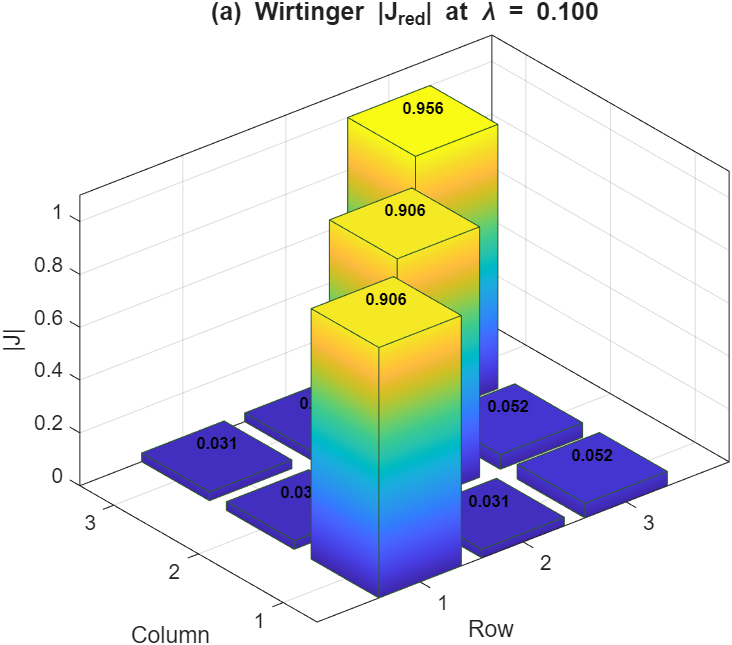}
\end{subfigure}
\hfill
\vspace{0.5cm} 
\begin{subfigure}{0.4\textwidth}
    \centering
    \includegraphics[width=\linewidth]{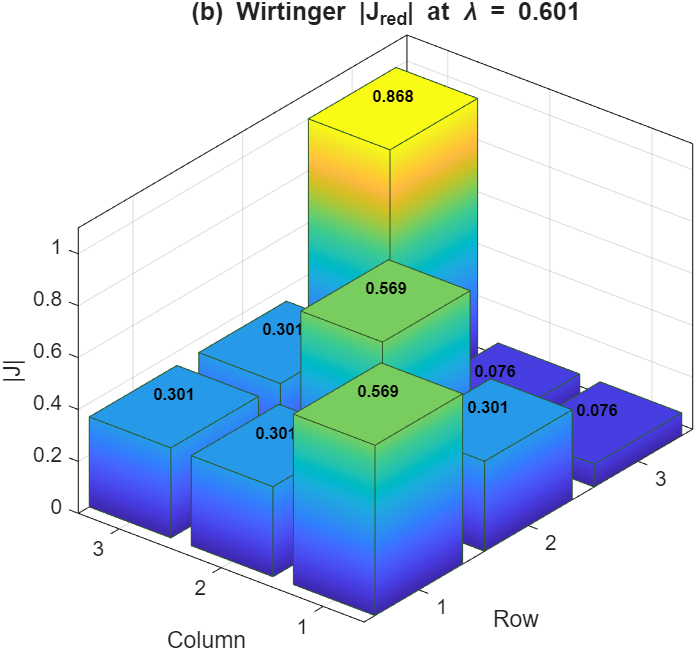}
\end{subfigure}
\caption{Evolution of the proposed Wirtinger Jacobian matrix across loading conditions}
\label{fig:wirtinger_evolution}
\end{figure}

Table~\ref{tab:stability_boundaries} compares the proposed index with established methods. The Wirtinger approach identifies the stability boundary between the conventional Jacobian and the more conservative L-index, which offers a meaningful safety margin without excessive conservatism. In real-time operation, this early warning capability provides an additional security buffer before the true voltage-collapse point.

\begin{table}[h]
\centering
\caption{Stability Boundary Comparison for the Three-Bus System}
\label{tab:stability_boundaries}
\begin{tabular}{lcc}
\hline
Method & Stability Boundary ($\lambda$) & Safety Margin (\%)\\
\hline
Proposed $C_W$ Index & 0.5974 & 0.59 \\
Conventional Jacobian       & 0.6009 & --   \\
L-index ($L \approx 0.8$)         & 0.5898 & 1.85 \\
\hline
\end{tabular}
\end{table}

\subsection{9-Bus IBR-Rich Test System}
\label{sec:9_test_system}

The proposed index is validated through power-flow simulations in MATPOWER~\cite{matpower-manual} on a representative network of the Hami region in China~\cite{xie2024oscillatory, Liu2017_SSO_Wind}, which includes both chain and 
radial configurations as shown in Fig.~\ref{fig:9bus}. The base power and base AC voltage are set to 
1000~MVA and 220~kV, respectively. Buses~1--6 are designated as PCCs interfaced with IBRs. The IBRs at 
PCCs~1, 2, 3, and~5 operate under PQ control with fixed reactive power 
injections, and are modeled as PQ buses whose voltage magnitudes vary 
with the power-flow solution. The IBRs at PCCs~4 and~6 operate under PV 
control, which regulates terminal voltages at fixed values. Buses~7 and~8 are 
interconnection nodes, and bus~9 is modeled as the slack bus.

\begin{figure}[t]
\centering
\includegraphics[width=0.85\linewidth]{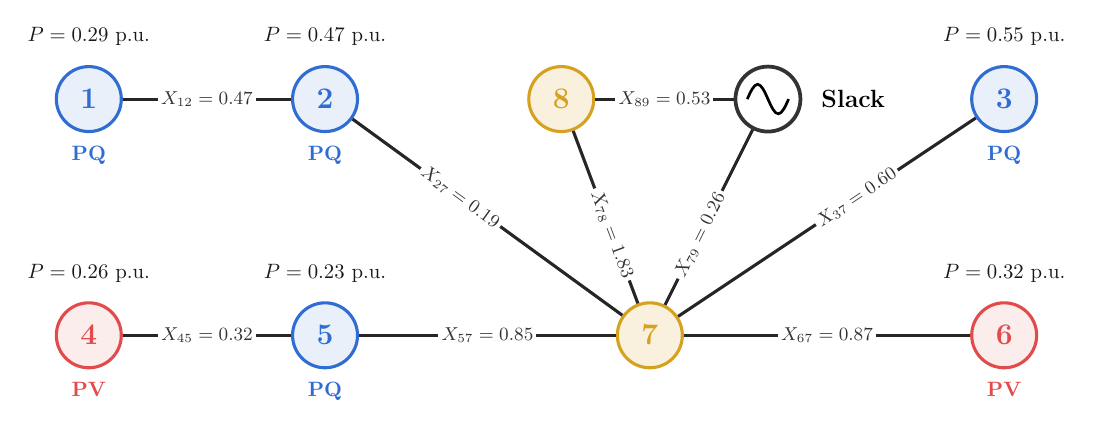}
\caption{
The 9-bus system}
\label{fig:9bus}
\end{figure}

\begin{figure}[t]
  \centering
  \includegraphics[width=\columnwidth]{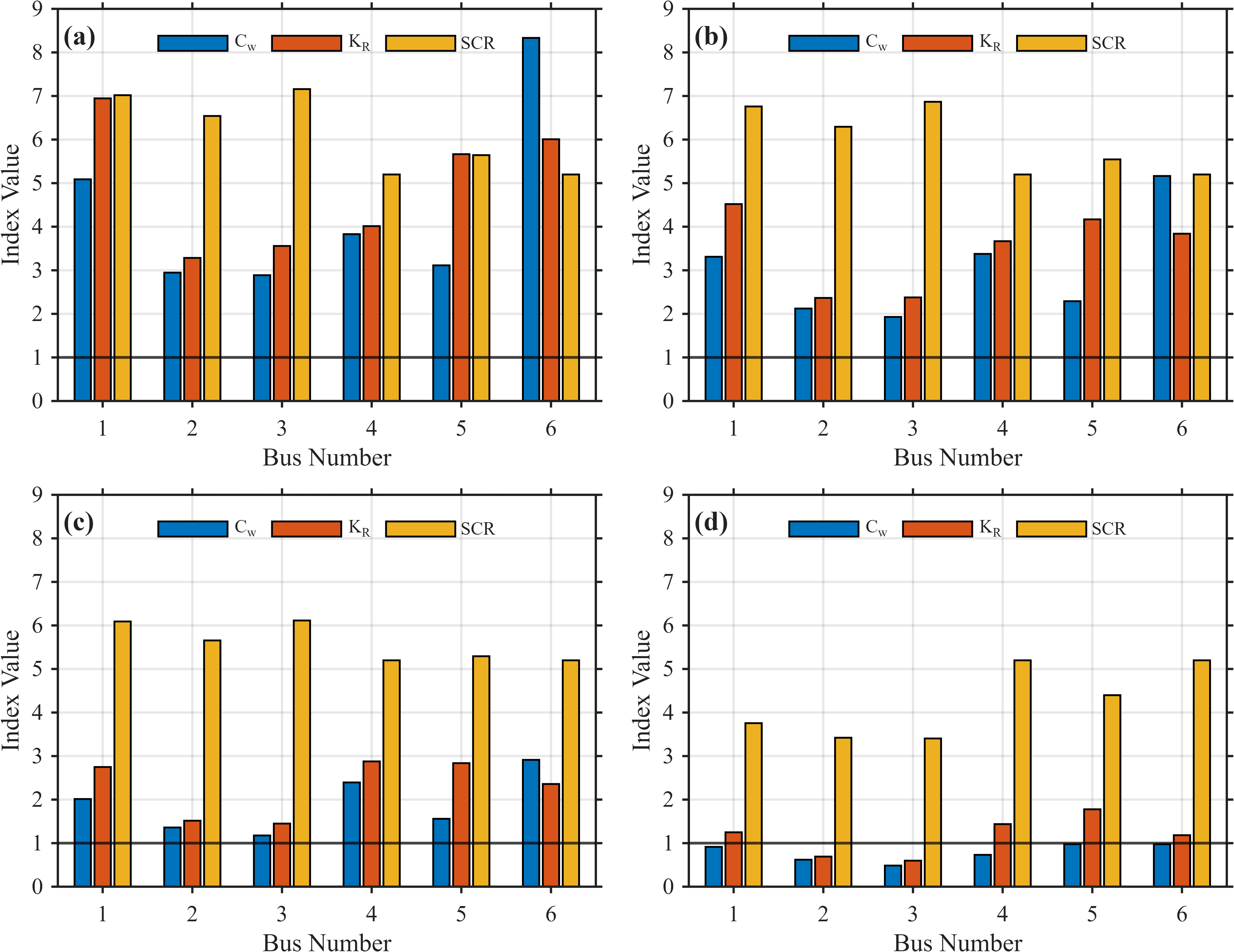}
  \caption{Comparison of the \(C_W\), \(K_{R}\)~\cite{Wang2024_Impedance}, and SCR indices across all buses at the critical operating point for the 9-bus system show loading levels: (a) 10\%, (b) 40\%, (c) 70\%, and (d) 100\%.}
  \label{fig:ninebus_indices}
\end{figure}

Fig.~\ref{fig:ninebus_indices} shows the evolution of all indices across 
loading levels. Under light loading, both $C_W$ and $K_{R}$ report 
large stability margins, as expected when the current injection is minimal 
and the system operates well within its voltage limits. As loading 
increases toward collapse, $C_W$ tends to decrease earlier than 
$K_{R}$, particularly at buses with strong inter-bus coupling. 
Specifically, $C_W$ decreases as the off-diagonal impedance coupling 
(row sums $\sum_j |Z_{ij}|$) grows relative to the diagonal terms, 
since this directly erodes the row-diagonal dominance of 
$J_{\mathrm{red}}$. In contrast, $K_{R}$ captures primarily local 
grid strength and is less sensitive to inter-bus coupling effects. 
Consequently, $C_W$ provides a more discriminative early warning under 
coupled network conditions.

Table~\ref{tab:Per-bus indices 9 bus} compares the per-bus indices at 
the critical loading point $P^\ast = 1.0$. Bus~3 exhibits the lowest 
$C_W$ value of (0.4853), which indicates the weakest voltage stability margin. This is consistent with its position in the chain configuration, where cumulative impedance along the feeder reduces the effective Thevenin voltage support relative to buses closer to the slack. The high $L$-index at Bus~3 ($0.5336$) corroborates this assessment, while the relatively low SCR ($3.4039$) confirms weak local grid strength at this terminal.

\begin{table}[h]
\centering
\caption{Per-bus voltage stability indices at critical loading 
         ($P^\ast = 1.0$).}
\label{tab:Per-bus indices 9 bus}
\begin{tabular}{c c c c c c}
\toprule
Bus & Type & $C_W$ & $K_{R}$ & SCR & $L$-index \\
\midrule
1 & $PQ$ & 0.9148 & 1.2489 & 3.7559 & 0.4165 \\
2 & $PQ$ & 0.6202 & 0.6911 & 3.4210 & 0.2277 \\
3 & $PQ$ & 0.4853 & 0.5981 & 3.4039 & 0.5336 \\
4 & $PV$ & 0.7316 & 1.4373 & 5.1987 & --     \\
5 & $PQ$ & 0.9767 & 1.7774 & 4.3982 & 0.0699 \\
6 & $PV$ & 0.9761 & 1.1826 & 5.1987 & --     \\
\bottomrule
\end{tabular}
\end{table}

\subsection{Modified IEEE 39-Bus System}
\label{sec:test_system}

\begin{figure}[t]
\centering
\includegraphics[width=0.9\linewidth]{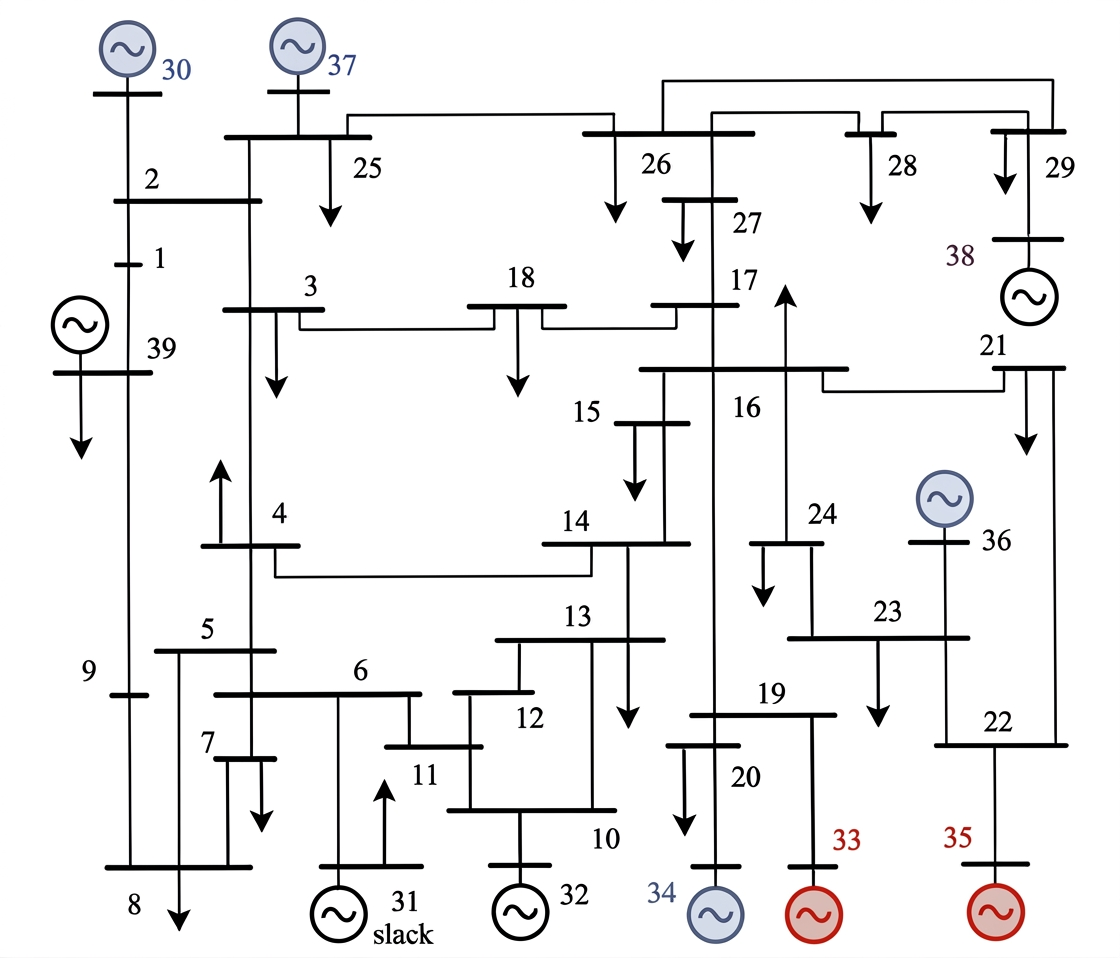}
\caption{
The IEEE 39-bus system}
\label{fig:39bus}
\end{figure}

The proposed index is evaluated on the IEEE 39-bus (New England) 
transmission system~\cite{Athay1979} shown in Fig.~\ref{fig:39bus}, with a base power of 
100~MVA and Bus~31 as the slack bus. Six  generator buses are replaced by IBR terminals:
\(
\mathcal{IBR} = \{30, 33, 34, 35, 36, 37\}.
\)
Active and reactive loads, along with generator real-power injections, are uniformly scaled to 30\% of their nominal values before any penetration sweep, which produces a well-conditioned base case in which subsequent loss of solvability is driven by increasing IBR injections.

The six IBR buses are partitioned into unconstrained and constrained 
sets:
\(
\mathcal{U} = \{30, 34, 36, 37\}, \quad
\mathcal{C} = \{33, 35\}.
\)
Buses in $\mathcal{U}$ are modeled as negative load injections representing grid-following converters with fixed active and reactive power output (marked in 
blue in Fig.~\ref{fig:39bus})~\cite{Rocabert2012}. Buses in $\mathcal{C}$ are voltage-controlled units representing grid-forming converters that regulate the terminal voltage~\cite{Tayyebi2020,NERC2021},
while the converter current remains below $I_{i,\max}$ (marked in red in Fig.~\ref{fig:39bus}). When $|I_i|$ approaches this limit, the converter saturates and transitions from voltage regulation to current-limited operation, wherein the current magnitude is held constant and its phase adjusts to maintain power 
balance. This transition is captured by the Wirtinger-based tangent factors discussed in Sec.~\ref{sec:tangent-framework}.

IBR injections are varied via a scalar $\lambda \in \{0.2, 0.4, 0.6, 0.8, 1.0\}$ that uniformly scales the active-power setpoints of all six units. At each converged operating point, $C_{W,i}$ is evaluated at all IBR buses, and the system margin is defined as $\min_i C_{W,i}$. For comparison, the $L$-index is computed at all 
load buses, SCR is obtained from Thévenin impedance ($Z_{\mathrm{th}}$) and local active power~\cite{Wu2018_SDSCR}, and the multi-infeed $K_R$ index is calculated from the maximum eigenvalue of the network impedance matrix~\cite{Wang2024_Impedance}.


\begin{figure}[t]
    \centering
    \includegraphics[width=\columnwidth]{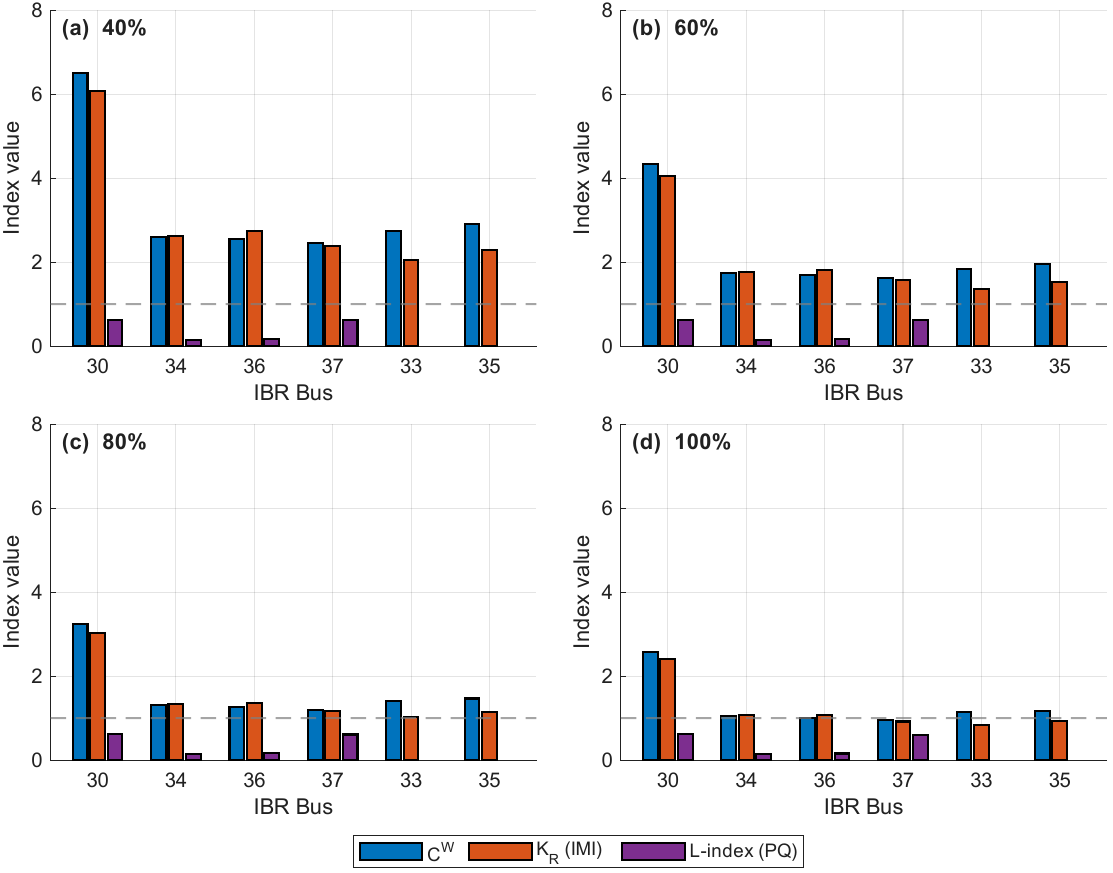}
    \caption{Comparative evolution of the unified index $C_W$, multi-infeed $K_R$ index, and $L$-index at the six IBR terminals for four loading levels.}
    \label{fig:barplots}
\end{figure}

\begin{figure}[t]
  \centering
  \includegraphics[width=\columnwidth]{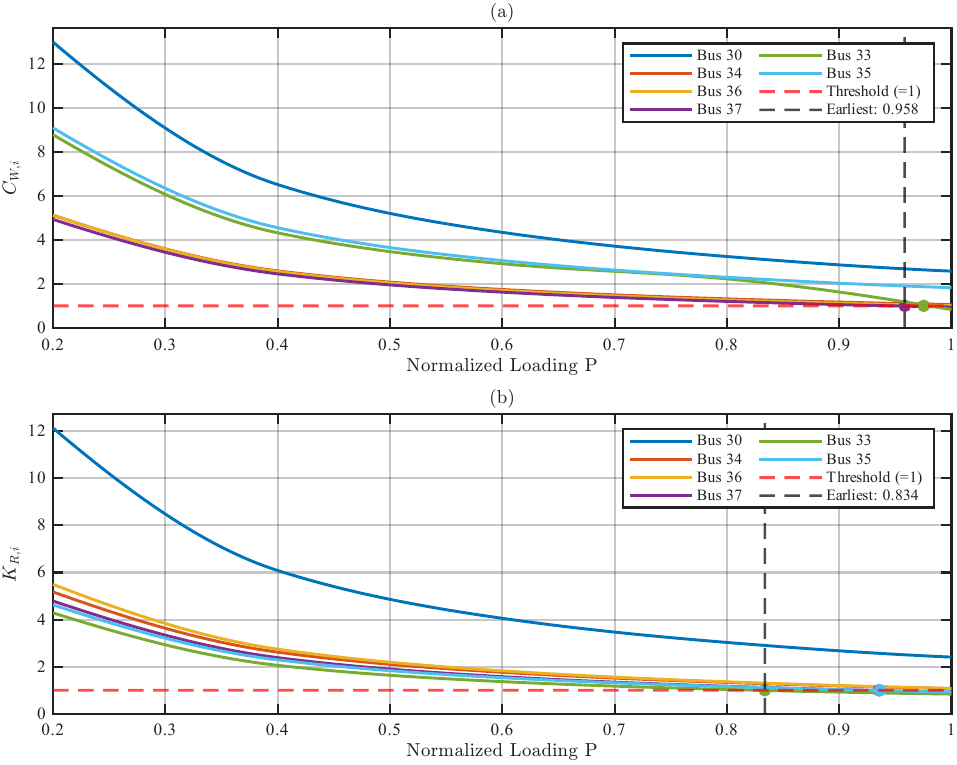}
  \caption{ (a) $C_{W}$ and (b) $K_{R}$ indices at all IBRs across loading sweep, with threshold crossings marked at $C_W = 1.0$ and $K_R = 1.0$}
  \label{fig:39bus_indices}
\end{figure}

Fig.~\ref{fig:barplots} summarizes the evolution of $C_W$, $K_R$, and the $L$-index across the six IBR terminals for four penetration levels. Both $C_W$ and $K_R$ decrease monotonically as the converter injections increase, which  reflects progressive depletion of the solvability margin. Up to $\lambda = 0.8$, bus~37 is the weakest location due to its large Thevenin impedance, which reaches $C_W = 1.20$ at $\lambda = 0.8$. 

The full per-bus results across all penetration cases are reported in Table~\ref{tab:cw_kr_l_scrop_all}. At $\lambda = 0.2$, $C_W$ ranges from 4.93 at bus~37 to 12.99 at bus~30, with all buses well within the solvability region. The spread across buses reflects differences in the Thevenin impedance seen at each PCC, with bus~30 exhibiting the strongest margin and bus~37 the weakest among the 
unconstrained terminals.

For moderate penetration ($\lambda \leq 0.8$), all operating points remain clearly stable with $C_W$ well above unity and no converter reaching its current limit. Near full penetration, $C_W$ reaches its critical value of 1.0 at $\lambda \simeq 0.96$, while $K_R = 1.0$ is reached earlier at $\lambda \simeq 0.8$, as shown in Fig.~\ref{fig:39bus_indices}.

The earlier threshold crossing of $K_R$ at 
$\lambda \simeq 0.80$ compared to $C_W$ at 
$\lambda \simeq 0.96$ reflects the fundamental difference between the two indices: $K_R$ applies a uniform threshold based on the maximum eigenvalue of the network impedance matrix~\cite{Wang2024_Impedance}, making it more conservative under partial network stress, whereas $C_W$ tracks the row-diagonal dominance of $J_{\mathrm{red}}$ at each bus independently. This makes $C_W$ less conservative and more discriminative at the bus level, particularly when only a subset of terminals approaches the solvability boundary.

At $\lambda = 1.0$, the converter at bus~33 reaches its current limit and transitions to current-limited operation. This causes $C_W$ to 
drop abruptly to 0.83, which shifts the most critical terminal from bus~37 to bus~33, as shown in Table~\ref{tab:Minimum_full_penetration}. In contrast, the $L$-index remains above 0.60 throughout, which confirms its insensitivity to instability driven by converter saturation rather than reactive-power deficiency. Furthermore, the SCR on bus~37 decreases from 22.6 to 4.5 as the injected power increases, which reflects a loss of relative grid strength due to growth in local generation rather than changes in network impedance; SCR alone cannot detect this distinction.

At bus~35, $C_W = 1.825$ while $K_R = 0.927$ at $\lambda = 1.0$, which represents the largest divergence between the two indices across all buses and penetration levels. Bus~35 operates under voltage-regulated ($\mathcal{C}^V$) mode throughout, and its converter has not yet reached its current limit. The $K_R$ index drops below unity at this bus due to the growing off-diagonal impedance coupling in the network, whereas $C_W$ correctly 
identifies bus~35 as still solvent since its local diagonal dominance condition remains satisfied. This contrast illustrates the bus-resolved advantage of $C_W$ over eigenvalue-based indices under mixed converter operating conditions.

\begin{table}[t]
\centering
\caption{Minimum indices at full IBR penetration level ($\lambda = 1.0$)}
\label{tab:Minimum_full_penetration}
\begin{tabular}{c c c c c c}
\hline
Bus & Type & $C_W$ & $K_R$ & $L$-index & $\mathrm{SCR}$ \\ \midrule
30 & $\mathcal{U}$   & 2.576 & 2.407 & 0.613 & 10.89 \\
34 & $\mathcal{U}$   & 1.060 & 1.072 & 0.147 & 4.55 \\
36 & $\mathcal{U}$   & 1.006 & 1.079 & 0.158 & 3.79 \\
37 & $\mathcal{U}$   & 0.947 & 0.920 & 0.607 & 4.52 \\
33 & $\mathcal{C}^{I}$ & \textbf{0.827} & 0.837 & -- & 4.22 \\
35 & $\mathcal{C}^{V}$ & 1.825 & 0.927 & -- & 3.81 \\ 
\bottomrule
\end{tabular}
\end{table} 

Table~\ref{tab:Minimum_full_penetration} confirms that the proposed index $C_W$ identifies bus~33 as the most critical terminal at $\lambda = 1.0$, with $C_W = 0.827$, while bus~37 also drops below unity ($C_W = 0.947$). Both $C_W$ and $K_R$ track the Jacobian conditioning deterioration more distinctly than the $L$-index or SCR, with $C_W$ providing the more precise bus-level resolution under mixed converter operating conditions.

\begin{table*}[t]
\centering
\caption{Unified $C_W$, multi-infeed $K_R$, $L$-index, and operating $\mathrm{SCR}$ for each IBR bus across all loading cases}
\label{tab:cw_kr_l_scrop_all}
\setlength{\tabcolsep}{2pt}
\scriptsize
\begin{tabular}{c|cccc|cccc|cccc|cccc|cccc}
\toprule
\multirow{2}{*}{Bus} &
\multicolumn{4}{c|}{Case 1 ($\lambda = 0.2$)} &
\multicolumn{4}{c|}{Case 2 ($\lambda = 0.4$)} &
\multicolumn{4}{c|}{Case 3 ($\lambda = 0.6$)} &
\multicolumn{4}{c|}{Case 4 ($\lambda = 0.8$)} &
\multicolumn{4}{c}{Case 5 ($\lambda = 1.0$)} \\
\cmidrule(lr){2-21}
 & $C_W$ & $K_R$ & $L$ & $\mathrm{SCR}$ &
   $C_W$ & $K_R$ & $L$ & $\mathrm{SCR}$ &
   $C_W$ & $K_R$ & $L$ & $\mathrm{SCR}$ &
   $C_W$ & $K_R$ & $L$ & $\mathrm{SCR}$ &
   $C_W$ & $K_R$ & $L$ & $\mathrm{SCR}$ \\
\midrule
30 ($\mathcal{U}$) &
12.99 & 12.13 & 0.613 & 54.44 &
6.51  &  6.08 & 0.614 & 27.22 &
4.34  &  4.05 & 0.616 & 18.15 &
3.24  &  3.03 & 0.619 & 13.61 &
2.58  &  2.41 & 0.622 & 10.89 \\

33 ($\mathcal{C}^{I}$) &
8.78 & 4.28 & -- & 21.09 &
4.32 & 2.06 & -- & 10.54 &
2.91 & 1.37 & -- & 7.03 &
2.22 & 1.03 & -- & 5.27 &
\textbf{0.83} & 0.84 & -- & 4.22 \\

34 ($\mathcal{U}$) &
5.12 & 5.17 & 0.147 & 22.74 &
2.59 & 2.62 & 0.145 & 11.37 &
1.74 & 1.76 & 0.144 & 7.58 &
1.32 & 1.33 & 0.143 & 5.69 &
1.06 & 1.07 & 0.142 & 4.55 \\

35 ($\mathcal{C}^{V}$) &
9.09 & 4.62 & -- & 19.03 &
4.55 & 2.29 & -- & 9.52 &
3.06 & 1.53 & -- & 6.34 &
2.30 & 1.16 & -- & 4.76 &
1.83 & 0.93 & -- & 3.81 \\

36 ($\mathcal{U}$) &
5.13 & 5.49 & 0.158 & 18.97 &
2.56 & 2.74 & 0.160 & 9.49 &
1.70 & 1.82 & 0.163 & 6.32 &
1.27 & 1.36 & 0.166 & 4.74 &
1.01 & 1.08 & 0.170 & 3.79 \\

37 ($\mathcal{U}$) &
4.93 & 4.79 & 0.607 & 22.61 &
2.45 & 2.38 & 0.610 & 11.31 &
1.62 & 1.58 & 0.614 & 7.54 &
1.20 & 1.17 & 0.619 & 5.65 &
0.95 & 0.92 & 0.624 & 4.52 \\

\bottomrule
\end{tabular}
\end{table*}

\section{Conclusion}
\label{sec:Conclusion}

A unified Wirtinger-based solvability index $C_W$ was developed 
for static voltage-stability assessment in power systems with high IBR penetration. The reduced Wirtinger Jacobian $J_{\mathrm{red}}$ accommodates unconstrained, voltage-regulated, 
and current-limited converter regimes through tangent-subspace constraint embedding, and its singularity set was proven to coincide exactly with that of the conventional power-flow Jacobian.

Case studies on benchmark systems demonstrate that the proposed $C_W$ index decreases monotonically with increasing converter penetration and accurately pinpoints the onset of current-limited operation at individual PCCs. The index closely tracks the deterioration of the conventional Jacobian's nonsingularity margin, whereas classical tools such as the $L$-index provide only weak variation and fail to signal emerging instability. Moreover, SCR-type strength measures remain high even when power-flow solvability is lost, confirming that voltage instability in converter-dominated grids stems primarily from converter state constraints rather than fault-level strength.

The results confirm that converter operating constraints fundamentally determine solvability margins in converter-rich grids. The proposed $C_W$ index captures this regime transition on a per-bus basis and in a computationally efficient manner, which provides system operators with more actionable stability certificates than traditional steady-state stability metrics under mixed operational constraints.

\appendices
\section{Equivalence of the Reduced Wirtinger and Conventional Jacobians}
\setcounter{equation}{0}  
\renewcommand{\theequation}{A\arabic{equation}}  
\label{app:eqv}

This appendix establishes the equivalence between the conventional (real-variable) power flow Jacobian and the reduced Wirtinger Jacobian for a general system comprising $n_u$ unconstrained and $n_c$ constrained buses.

\subsection{Row and Column Maps}
Define the \emph{row map} $L$ that transforms complex power variations into their real counterparts:
\begin{equation}
\label{eq:row-map}
\begin{aligned}
\underbrace{\begin{bmatrix} \Delta P_{\mathcal{U}} \\[0.2ex] \Delta Q_{\mathcal{U}}  \\[0.2ex] \Delta P_{\mathcal{C}} \end{bmatrix}}_{\text{real powers}}
&=
\underbrace{\begin{bmatrix}
\frac{1}{2}I_{n_u} & \frac{1}{2}I_{n_u} & 0 \\
\frac{1}{2\jj}I_{n_u} & -\frac{1}{2\jj}I_{n_u} & 0 \\
0 & 0 & I_{n_c}
\end{bmatrix}}_{=:L}\,
\underbrace{\begin{bmatrix} \Delta S_{\mathcal{U}} \\[0.2ex] \Delta S_{\mathcal{U}}^{*} \\[0.2ex] \Delta P_{\mathcal{C}} \end{bmatrix}}_{\text{complex inputs}},
\\[0.6ex]
\text{equivalently}\quad
L &=
\begin{bmatrix}
K\!\otimes\! I_{n_u} & 0 \\[0.3ex]
0 & I_{n_c}
\end{bmatrix},
\qquad
K \coloneqq
\begin{bmatrix}
\frac{1}{2} & \frac{1}{2} \\
\frac{1}{2\jj} & -\frac{1}{2\jj}
\end{bmatrix},
\\[0.6ex]
\det L &= (\det K)^{n_u} = \left(-\frac{1}{2j}\right)^{n_u}\neq 0.
\end{aligned}
\end{equation}

Let $Y_{\mathrm{red}}:=Z^{-1}$ be the reduced admittance matrix, and define the mapping between voltage perturbations and real polar state increments:
\begin{equation}
\begin{bmatrix} \Delta V_{\mathcal{U}} \\ \Delta V_{\mathcal{C}} \end{bmatrix}
=
\underbrace{\begin{bmatrix}
\Diag{e^{\jj\theta_{\mathcal{U}}}} & 0                 & \jj\,\Diag{V_{\mathcal{U}}}\\
0                      & \jj\,\Diag{V_{\mathcal{C}}}   & 0
\end{bmatrix}}_{=:M}
\begin{bmatrix} \Delta U_{\mathcal{U}} \\ \Delta \theta_{\mathcal{C}} \\ \Delta \theta_{\mathcal{U}} \end{bmatrix}.
\end{equation}

Then,
\begin{equation}
\begin{bmatrix} \Delta I_{\mathcal{U}} \\ \Delta I_{\mathcal{C}} \end{bmatrix}
= Y_{\mathrm{red}}
\begin{bmatrix} \Delta V_{\mathcal{U}} \\ \Delta V_{\mathcal{C}} \end{bmatrix}
= Y_{\mathrm{red}}\, M
\begin{bmatrix} \Delta U_{\mathcal{U}} \\ \Delta \theta_{\mathcal{C}} \\ \Delta \theta_{\mathcal{U}} \end{bmatrix}.
\end{equation}

Stacking $\Delta I_{\mathcal{U}}, \Delta I_{\mathcal{U}}^{*}, \Delta I_{\mathcal{C}}$ yields the \emph{column map} from real polar states to Wirtinger current increments:
\begin{equation}
\underbrace{\begin{bmatrix} \Delta I_{\mathcal{U}} \\ \Delta I_{\mathcal{U}}^{*} \\ \Delta I_{\mathcal{C}} \end{bmatrix}}_{\Delta x_W}
=
\underbrace{\begin{bmatrix}
\frac{\partial I_{\mathcal{U}}}{\partial U_{\mathcal{U}}} & \frac{\partial I_{\mathcal{U}}}{\partial \theta_{\mathcal{C}}} & \frac{\partial I_{\mathcal{U}}}{\partial \theta_{\mathcal{U}}} \\
\big(\frac{\partial I_{\mathcal{U}}}{\partial U_{\mathcal{U}}}\big)^{*} & \big(\frac{\partial I_{\mathcal{U}}}{\partial \theta_{\mathcal{C}}}\big)^{*} & \big(\frac{\partial I_{\mathcal{U}}}{\partial \theta_{\mathcal{U}}}\big)^{*} \\
\frac{\partial I_{\mathcal{C}}}{\partial U_{\mathcal{U}}} & \frac{\partial I_{\mathcal{C}}}{\partial \theta_{\mathcal{C}}} & \frac{\partial I_{\mathcal{C}}}{\partial \theta_{\mathcal{U}}}
\end{bmatrix}}_{=:R}
\underbrace{\begin{bmatrix} \Delta U_{\mathcal{U}} \\ \Delta \theta_{\mathcal{C}} \\ \Delta \theta_{\mathcal{U}} \end{bmatrix}}_{\Delta x_{\mathrm{conv}}},
\label{eq:R-map}
\end{equation}
with the blocks read directly from $Y_{\mathrm{red}}M$:
\[
\begin{aligned}
&\frac{\partial I_{\mathcal{U}}}{\partial U_{\mathcal{U}}} = Y_{\mathcal{UU}}\,\Diag{e^{\jj\theta_{\mathcal{U}}}}, \quad
\frac{\partial I_{\mathcal{U}}}{\partial \theta_{\mathcal{U}}} = \jj\,Y_{\mathcal{UU}}\,\Diag{V_{\mathcal{U}}}, \\[3pt]
&\frac{\partial I_{\mathcal{C}}}{\partial U_{\mathcal{U}}} = Y_{\mathcal{CU}}\,\Diag{e^{\jj\theta_{\mathcal{U}}}}, \quad
\frac{\partial I_{\mathcal{C}}}{\partial \theta_{\mathcal{U}}} = \jj\,Y_{\mathcal{CU}}\,\Diag{V_{\mathcal{U}}}, \\[3pt]
&\frac{\partial I_{\mathcal{U}}}{\partial \theta_{\mathcal{C}}} = \jj\,Y_{\mathcal{UC}}\,\Diag{V_{\mathcal{C}}}, \quad
\frac{\partial I_{\mathcal{C}}}{\partial \theta_{\mathcal{C}}} = \jj\,Y_{\mathcal{CC}}\,\Diag{V_{\mathcal{C}}}.
\end{aligned} 
\]

Here \(dV_{\mathcal{U}}=\D{e^{\jj\theta_{\mathcal{U}}}}\,dU_{\mathcal{U}}+\jj\,\D{V_{\mathcal{U}}}\,d\theta_{\mathcal{U}}\), \(dV_{\mathcal{C}}=\jj\,\D{V_{\mathcal{C}}}\,d\theta_{\mathcal{C}}\), and \([dI_{\mathcal{U}};dI_{\mathcal{C}}]=Y_{\mathrm{red}}[dV_{\mathcal{U}};dV_{\mathcal{C}}]\). At a regular operating point, \(R\) is nonsingular.

\subsection{Equivalence Theorem}
\begin{theorem}[Identical singular sets]
\label{thm:equivalence}
Let \(J_{\mathrm{conv}}=\partial [P_{\mathcal{U}};P_{\mathcal{C}};Q_{\mathcal{U}}]/\partial[\theta_{\mathcal{U}};\theta_{\mathcal{C}};U_{\mathcal{U}}]\) be the conventional Jacobian and \(J_{\mathrm{red}}=\partial [S_{\mathcal{U}};S_{\mathcal{U}}^{*};P_{\mathcal{C}}]/\partial[I_{\mathcal{U}};I_{\mathcal{U}}^{*};I_{\mathcal{C}}]\) the reduced Wirtinger Jacobian in \eqref{eq:Jred-general}. Then
\begin{equation}
\label{eq:chain}
J_{\mathrm{conv}} \;=\; L\,J_{\mathrm{red}}\,R,
\qquad
\det J_{\mathrm{conv}}=0 \;\Longleftrightarrow\; \det J_{\mathrm{red}}=0,
\end{equation}
whenever \(L\) and \(R\) are nonsingular (which holds generically at regular operating points).
\end{theorem}

\begin{IEEEproof}[Proof]
By construction, \(\Delta F_{\mathrm{conv}}=L\,\Delta F_{\mathrm{red}}\) (row map) and \(\Delta x_{\mathrm{red}}=R\,\Delta x_{\mathrm{conv}}\) (column map). Chain rule yields \(J_{\mathrm{conv}}=L\,J_{\mathrm{red}}\,R\). Since \(\det(L)\neq 0\) and \(R\) is nonsingular away from pathological points (e.g., zero voltages or degenerate admittances), \(J_{\mathrm{conv}}\) and \(J_{\mathrm{red}}\) have identical singular sets.
\end{IEEEproof}

\begin{remark}
The use of Wirtinger derivatives to construct complex-valued Jacobians that are equivalent (up to nonsingular row/column transformations) to the conventional real-variable Jacobian is consistent with prior work on Wirtinger-based load flow and optimal power flow
\cite{garces2019wirtingerpf,dvzafic2018high},
and with the general Wirtinger calculus framework
\cite{brandwood1983complexgradient,kreutzdelgado2009crcalculus,hjorungnes2011complex}.
Hence, the relation~\eqref{eq:chain} in
Theorem~\ref{thm:equivalence} preserves the classical Jacobian
singularity boundary used in voltage stability analysis
\cite{sauer1990jacobian,ajjarapu1992cpf,gao1992modal}.
\end{remark}

\bibliographystyle{IEEEtran}
\bibliography{references}

\end{document}